\documentclass[journal]{IEEEtran}
\usepackage{amsmath,amsthm,amssymb,amsfonts}
\usepackage{mathtools}
\usepackage{mathrsfs}
\usepackage{algorithmic}
\usepackage[ruled,linesnumbered]{algorithm2e}
\usepackage{array}
\usepackage[caption=false,font=tiny,labelfont=sf,textfont=sf]{subfig}
\usepackage{tabularx}
\usepackage{multirow}
\usepackage{textcomp}
\usepackage{stfloats}
\usepackage{url}
\usepackage{verbatim}
\usepackage[pdftex]{graphicx,xcolor}
\usepackage{cite}
\usepackage{cleveref}
\usepackage{newtxtext}
\usepackage[varg]{newtxmath}
\usepackage{comment}

\crefrangeformat{equation}{(#1)-(#2)}
\newtheorem{theorem}{Theorem}

\title{Covert Multicast in UAV-Enabled Wireless Communication Systems With One-hop and Two-hop Strategies}

\author{Wenhao Zhang\textsuperscript{1}, Ji He\textsuperscript{2,*}, Yuanyu Zhang\textsuperscript{2}, \\
\IEEEauthorblockA{1. School of Systems Information Science, Future University Hakodate, Hakodate, Japan}

\IEEEauthorblockA{2. School of Computer Science and Technology, Xidian University, Xi'an, China}

\IEEEauthorblockA{* Corresponding author: garyhej1991@gmail.com}
}

\begin{document}

\maketitle

\begin{abstract}
This paper delves into the time-efficient covert multicast in a wireless communication system facilitated by Unmanned Aerial Vehicle (UAV), in which the UAV aims to disseminate a common covert information to multiple ground users (GUs) while suffering from the risk of detection by a ground warden (Willie). We
propose one hop (OH) and two hop (TH) transmission schemes,
first develop a theoretical framework for performance modeling of both the detection error probability at Willie and the transmission time at UAV. The optimization problems subject to the covertness constraint for the two transmission schemes are then formulated to gain insights into the system settings of the UAV's prior transmit probability, transmit power and horizontal location that affect the minimum transmission time. The optimization problems are non-convex and challenging to give numerical results. We thus explore the optimal setting of the transmit power and the prior transmit probability for the UAV separately under specific parameters with two schemes. We further propose a particle swarm optimization (PSO) based algorithm and an exhaustive algorithm to provide the joint solutions for the optimization problem with the OH transmission scheme and TH scheme, respectively. Finally, the efficiency of the proposed PSO-based algorithm is substantiated through extensive numerical results.

\end{abstract}


\section{Introduction}
Unmanned aerial vehicle (UAV) assisted communication has the advantages of high mobility, swift deployment, and enhanced channel conditions, making it rapidly grow and widely employed in various wireless communication systems, including civilian and military. With ongoing cost reductions and endurance improvements of UAVs, UAV-enabled communications are promising to play an increasingly important role in future wireless systems. Thus attracting significant attention from both industry and academia \cite{zeng2019accessing,wu2018cooperative}. However, the broadcast nature of wireless channels makes UAV-enabled communication susceptible to malicious eavesdroppers or wardens \cite{cui2018robust}. Covert communication has recently been regarded as an excellent solution against these attacks. Unlike conventional security approaches that focus solely on protecting message content, covert communication aims to hide the communication process itself, providing a higher level of security \cite{bash2013limits}.

Currently, extensive contributions have been devoted to the study of UAV-enabled covert wireless communication systems. For a UAV system where the UAV acts as the transmitter, the authors in \cite{zhou2019joint} propose an iterative algorithm to efficiently optimize the UAV's trajectory and transmit power concurrently so as to achieve the maximal covert communication rate in the system. In  \cite{rao2022optimal}, a geometric method algorithm is proposed to jointly determine the optimal UAV trajectory and transmit power for the covert communication capacity maximization in a UAV system. The authors in \cite{jiang2021covert} present an iterative algorithm to collectively optimize the time slot, transmit power and trajectory of UAV for the covertness performance enhancement in a UAV system. The work in \cite{yan2021optimal} and \cite{zhou2021three} delves into the optimization of UAV transmit power and placement for covert rate maximization in two-dimensional and three-dimensional UAV systems. The work in \cite{chen2021uav} focus on scenarios where the UAV acts as a relay, optimizing block length and transmit power at both the transmitter and UAV to enhance the covert performance. Regarding the full-duplex UAV relay system, the authors in \cite{zhang2022uav} formulate an optimization problem that jointly optimizes the transmit power of the transmitter and jamming power of the UAV relay. They develop a penalty successive convex approximation scheme to tackle this optimization problem and achieve the maximum covert throughput. Moreover, considering the delay-intolerant UAV relay system, the authors in \cite{jiao2022placement} explore UAV location optimization to maximize the effective throughput under the constraints of covertness and block length.

In light of the considerable advantage of intelligent reflection surface (IRS) technology in reconfiguring the propagation environment, the authors in \cite{chen2023uav} propose a novel scheme to covert communication by integrating UAVs and IRS. In this scheme, the UAV carries an IRS to introduce additional randomness to channels and the optimal transmit-to-jamming ratio is obtained to maximize the achievable covert transmission rate. The authors in \cite{ wang2022covert} iteratively optimize the transmit power of the transmitter, the phase shift of the IRS, and the horizontal location of the UAV-IRS to maximize the covert transmission rate. The authors in \cite{qian2023joint} develop an iterative algorithm aiming to optimize the UAV 2D trajectory and IRS phase shift to maximize the average covert rate in a UAV-enabled wireless communication system. Furthermore, in \cite{bi2023deep}, a deep reinforcement learning-based optimization algorithm is proposed for optimizing the UAV 3D trajectory and IRS phase shift to enhance covert communication performance. Concerning multi-user systems, the authors in \cite{tatar2022aerial} devise an energy-efficient multi-UAV-mounted IRS covert communication scheme over the Terahertz band. They propose a block successive convex approximation approach to iteratively design the user scheduling, power allocation, maximum jamming power, IRS beamforming, and UAVs' trajectory to improve covert throughput. Besides, the non-orthogonal multiple access (NOMA), which can facilitate multi-users to access services simultaneously over the same frequency, offers enhanced throughput to multi-user systems. Consider the UAV-assisted NOMA networks, the authors in \cite{su2023optimal} jointly optimize the hovering height and power to enhance the covert performance. In addition, the authors in \cite{deng2023joint} propose an iterative block coordinate descent-based successive convex approximation method to maximize the average covert achievable rate for the considered UAV-assisted NOMA networks.

Observing that multicast communication stands as a crucial communication paradigm for enhancing transmission efficiency and overall network utility \cite{sidiropoulos2006transmit}, especially in effectively utilizing network radio resources to transmit common data to multiple users simultaneously. The multicast communication holds promise for supporting the future content-centric applications \cite{lin2021supporting}, and it has significant potential to enhance the efficiency of covert communication systems. Thus, how to implement the efficient covert multicast communication becomes an issue of considerable importance.

Notice that existing covert communication primarily focuses on unicast transmission between transceivers. Covert multicast communication remains an unexplored issue. To fill up this gap, this paper considers a UAV-enabled multicast communication system, where the UAV acts as a flying base station, disseminating common information to multiple ground users (GUs) while avoiding detection by a ground warden. Recognize that multicast performance is fundamentally limited by the worst communication link of the transmitter to the receiver. The UAV-enabled multicast system mitigates this limitation by adaptive location design of the UAV to balance the overall communication link, thus enhancing transmission performance. Under the above considerations, this paper aims to seek the optimal transmit power and hover location (i.e., horizontal location) of the UAV to minimize the overall covert transmission time. The main contributions of this paper are summarized as follows.

\begin{itemize}
\item We first develop a theoretical framework encompassing the performance modeling of both the detection error probability at Willie and the overall transmission time of the UAV. We also formulate the optimization problems to determine the optimal settings for the transmit power and hover location of the UAV, aiming to minimize the transmission time while satisfying the covertness and altitude constraint of the UAV with one hop (OH) and two hop (TH) transmission schemes.

\item We then explore the mathematical results of the optimal transmit power, optimal prior transmit probability and corresponding minimum transmission time with the given hover location of the UAV. We further propose a particle swarm optimization (PSO)-based optimization algorithm and an exhaustive algorithm that jointly optimizes the UAV’s transmit power, prior transmit probability and horizontal location to achieve the overall minimum transmission time for the OH transmission scheme and TH scheme, respectively.

\item Finally, we provide extensive numerical results under various GUs distribution scenarios to illustrate the comparison between the optimal transmit power setting with fixed-location of the UAV and the proposed optimization algorithms, and thus to demonstrate the efficiency of the proposed algorithms.

\end{itemize}

The remainder of this paper is organized as follows. Section~\ref{sec_sys_model} introduces the preliminaries, and Sections ~\ref{sec_tran_stra} and ~\ref{sec_tran_stra_th} present the detection error probability (DEP) analysis at Willie, transmission time minimization under OH transmission protocol and TH transmission protocol, respectively. The numerical results are provided in Section~\ref{sec_num}. Finally, Section~\ref{sec_con} concludes this work.

\section{Preliminaries} \label{sec_sys_model}
\subsection{System Model}
As shown in Fig.~\ref{fig_sys_model}, the UAV (Alice, $a$) acts as a flying base station in a three-dimensional area with the altitude $h$. It disseminates a common covert information (CI) of $M$ bits to a total size of $G$ ground users (GUs), denoted by the set $\mathbb{G}=\{1, 2, \cdots g \cdots, G \}$. Meanwhile, a ground warden (Willie, $w$) wants to detect the transmission. In this system, GUs are randomly distributed following the Poisson point process with a density of $\lambda$ within a circular area of radius $r$ covered by the UAV. At the same time, Willie is situated outside of this circular area, at a distance of $d_w$ from the center of this circle. Let $\mathbf{q}_{a}=[x_{a},y_{a}]$ denotes the horizontal coordinate of the UAV, and the horizontal coordinates of the $g$-th GU and Willie are denoted as $\mathbf{q}_{g}=[x_{g},y_{g}]$ and $\mathbf{q}_{w}=[x_{w},y_{w}]$, respectively. Thus, the horizontal distance from the UAV to the $g$-th GU and the distance from the UAV to Willie are given by $l_{ag}=\Vert\mathbf{q}_a-\mathbf{q}_g\Vert$ and $l_{aw}=\Vert\mathbf{q}_a-\mathbf{q}_w\Vert$, respectively. The distance from the UAV to the $g$-th GU and the distance from the UAV to Willie are given by $d_{ag}=\sqrt{l_{ag}^2+h^2}$ and $d_{aw}=\sqrt{l_{aw}^2+h^2}$, respectively.
\begin{figure}[tb]
\centering
\includegraphics[width=0.45\textwidth]{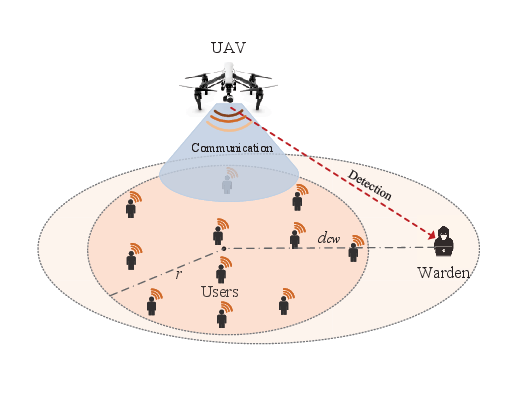}
\caption{Illustration of the system model.}
\label{fig_sys_model}
\end{figure}

\subsection{Channel Model}
\subsubsection{Air to Ground Channel Model}
The air-to-ground (A2G) channels from the UAV to ground nodes (i.e., GUs and Willie) are modeled as the probabilistic line-of-sight (LoS) channels. Following the model presented in \cite{yan2021optimal}, the probability of the LoS channel from the UAV to the $g$-th GU and Willie can be respectively given by
\begin{equation}\label{LoS_pro_ag}
\begin{aligned}
P_{LoS}(\theta_{g})=\frac{1}{1+e\exp{(-f[\theta_{g}-e])}},
\end{aligned}
\end{equation} 
and 
\begin{equation}\label{LoS_pro_aw}
\begin{aligned}
P_{LoS}(\theta_{w})=\frac{1}{1+e\exp{(-f[\theta_{w}-e])}},
\end{aligned}
\end{equation} 
where $e$ and $f$ represent the S-curve parameters related the specific communication environment; $\theta_{g}=\frac{180^{\circ}}{\pi}\arctan(\frac{h}{l_{ag}})$ and $\theta_{w}=\frac{180^{\circ}}{\pi}\arctan(\frac{h}{l_{aw}})$ are the degree of the elevation angles for the $g$-th GU relative to the UAV and Willie relative to the UAV, respectively. Therefore, as per \cite{shu2019delay, yan2021optimal,su2023optimal}, the channel gain $H_{ag}$ from the UAV to the $g$-th GU can be given by
\begin{equation}\label{channel_gain_ag}
\begin{aligned}
H_{ag}= P_{LoS}(\theta_{g})d_{ag}^{\alpha_{L}},
\end{aligned}
\end{equation}
where $\alpha_{L}<0$ is the A2G LoS channel fading loss exponent. The channel gain $H_{aw}$ from the UAV to Willie can be given by
\begin{equation}\label{channel_gain_aw}
\begin{aligned}
H_{aw}= P_{LoS}(\theta_{w})d_{aw}^{\alpha_{L}}.
\end{aligned}
\end{equation}

\subsubsection{Ground to Ground Channel Model} 
The ground-to-ground (G2G) channels among GUs and Willie are assumed to follow the quasi-static Rayleigh fading channel model, where each channel remains constant in one slot while changing independently and randomly from one slot to another according to Rayleigh distribution with zero mean and unit variance. Suppose that a GU is acting as a relay ($r$), the channel from the relay to the $g$-th GU is denoted by $h_{rg}$ ($r,g \in \mathbb{G}, r\ne g$), the associated channel gain $H_{rg}$ can be given by
\begin{equation}\label{channel_gain_rg}
\begin{aligned}
H_{rg}= \vert h_{rg}\vert ^{2}l_{rg}^{\alpha_{N}},
\end{aligned}
\end{equation}
where $l_{rg}$ is the distance from the relay to the $g$-th GU and $\alpha_{N}<0$ is the G2G channel fading loss exponent. The channel from the relay to Willie is denoted by $h_{rw}$, the associated channel gain $H_{rw}$ can be given by
\begin{equation}\label{channel_gain_rw}
\begin{aligned}
H_{rw}= \vert h_{rw}\vert ^{2}l_{rw}^{\alpha_{N}},
\end{aligned}
\end{equation}
where $l_{rw}$ is the distance from the relay to Willie. 

In this paper, we assume the UAV, GUs and Willie only have the knowledge of the statistical characterizations of the channel state information.

\subsection{Transmission Scheme}
In this paper, we design two covert multicast transmission schemes for the concerned system based on one-hop (OH) transmission and two-hop (TH) transmission, respectively. 

\subsubsection{OH Transmission Scheme}
Under the OH transmission scheme, given the hover location of the UAV, the UAV directly multicasts CI to the GUs. As illustrate in fig.~\ref{fig_comm_model_oh}, the UAV first files to the target location and divides the CI into $m$ blocks. To confuse Willie, the UAV selects each time slot with probability $\pi_{1}$ (i.e., does not select this time slot with probability $\pi_{0}=1-\pi_{1}$) to transmit a block. If the UAV decides to transmit one block in a time slot, she encodes this block with a secret Gaussian codebook pre-shared among the UAV and GUs to ensure reliable transmission. This block is thus encoded into $n$ real value symbols $\mathbf{x}=(\mathbf{x}[1],\cdots,\mathbf{x}[i])$, where each symbol satisfies $\mathbf{x}[i]\sim \mathcal{N} (0, 1)$ and $i=1,2,\cdots,n$. Finally, she multicasts each symbol $\mathbf{x}[i]$ to the GUs. Overall, the received signal at the $g$-th GU for the $i$-th symbol under the OH transmission scheme can be given by
\begin{equation}
\begin{aligned} \label{rec_GU}
\mathbf{y}_{g}[i]=\sqrt{P_{a}H_{ag}}\mathbf{x}[i] +\mathbf{n}_{g}[i],
\end{aligned}
\end{equation}
where $P_{a}$ is the transmit power of the UAV; $\mathbf{n}_{g}[i] \sim  \mathcal{N} (0, \sigma_{g}^{2})$ is the additive white Gaussian noise (AWGN) with mean zero and variance $\sigma_{g}^{2}$ at the $g$-th GU. 
\begin{figure}[tb]
\centering
\includegraphics[width=0.45\textwidth]{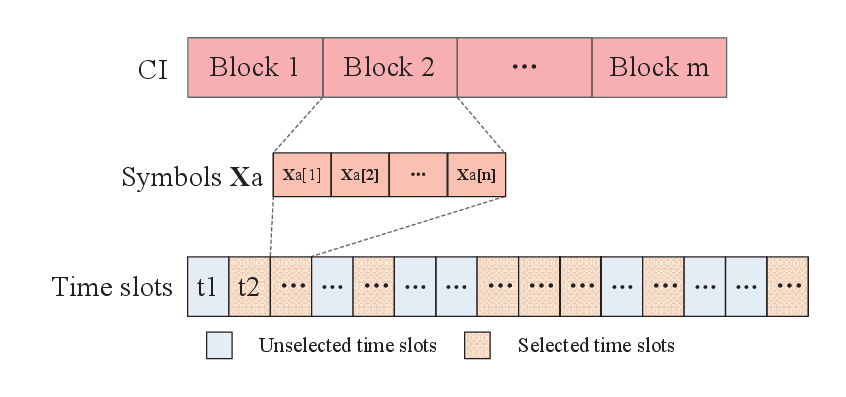}
\caption{Illustration of the OH Transmission Scheme.}
\label{fig_comm_model_oh}
\end{figure}

\subsubsection{TH Transmission Scheme}
Under the TH transmission scheme, a GU is selected to act as the relay, the UAV transmits CI to this relay, and the relay then multicasts CI to the remaining GUs. Specifically, the UAV flies to the location just above the relay, then similar to OH transmission protocol, as shown in fig.~\ref{fig_comm_model_th}, she divides the CI into blocks, selects a time slot with probability $\pi_{1}$, and encodes each block into $n$ real value symbols. Instead of multicasting CI to the GUs of the OH transmission scheme, the UAV transmits symbols of each block to the relay in a selected time slot. The relay then decodes the signals and re-encodes them into symbols using the same codebook and approach as the UAV. After that, the relay forwards these symbols of one block to the remaining GUs in a multicast manner during the following time slot when the UAV does not transmit information. To achieve this, each GU is equipped with two independent antennas for reception and transmission, enabling it to operate in half-duplex (HD) transmission mode. Additionally, the relay acts as the buffer-aided relay to store the symbols received from the UAV. The buffer queues follow the first-in-first-out (FIFO) principle. Moreover, the relay operates in the decode-and-forward (DF) mode with no processing delay \cite{wang2018covert}. Overall, under the TH transmission scheme, for the $i$-th symbol, the received signals at the relay $\mathbf{y}_{r}[i]$, the $g$-th remaining GU when the UAV transmits CI $\mathbf{y}_{g}^{u}[i]$, and the $g$-th remaining GU when the relay transmits CI $\mathbf{y}_{g}^{r}[i]$ are respectively given by 
\begin{align} 
&\mathbf{y}_{r}[i]=\sqrt{P_{a}H_{ar}}\mathbf{x}[i] +\mathbf{n}_{r}[i],\label{rec_relay}\\
&\mathbf{y}_{g}^{u}[i]=\sqrt{P_{a}H_{ag}}\mathbf{x}[i] + \mathbf{n}_{g}[i], \label{rec_relay_u_g}\\
&\mathbf{y}_{g}^{r}[i]=\sqrt{P_{r}H_{rg}}\mathbf{t}[i] + \mathbf{n}_{g}[i],\label{rec_relay_r_g}
\end{align}
where $H_{ar}$ is similar to the equation (\ref{channel_gain_ag}), denotes the channel gain from the UAV to the relay; $\mathbf{n}_{r}[i] \sim  \mathcal{N} (0, \sigma_{r}^{2})$ is the AWGN at the relay; $P_{r}$ is the predetermined transmit power at the relay; $\mathbf{t}[i]=\mathbf{x}[i]$ is the transmit symbols from the relay.

\begin{figure}[tb]
\centering
\includegraphics[width=0.45\textwidth]{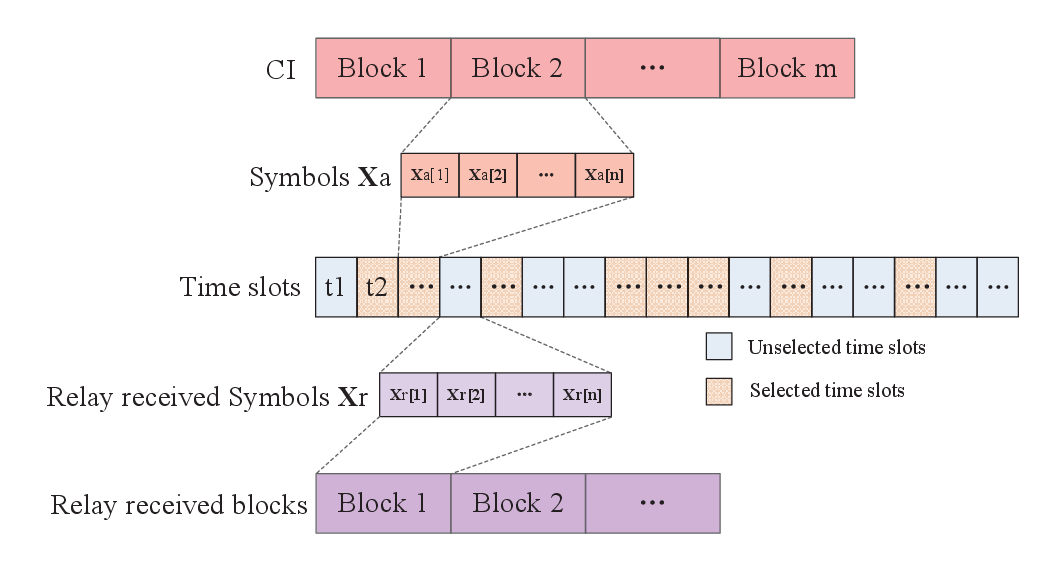}
\caption{Illustration of the TH Transmission Scheme.}
\label{fig_comm_model_th}
\end{figure}
 
\subsection{Detection Model at Willie}
Based on the observations over all time slots, Willie attempts to determine whether the transmission happened or not at the UAV. To achieve this, she adopts the hypothesis test to distinguish between the null hypothesis $H_0$, which infers that the UAV did not conduct transmission, and the alternative hypothesis $H_1$, which suggests that the UAV did transmit. The detection models of Willie under the two transmission protocols are as follows. 

\subsubsection{OH Transmission Scheme}
To detect the transmission of UAV to GUs under OH transmission scheme, Willie adopts the hypothesis test and distinguishes the following two received signals.
\begin{equation}
\begin{aligned} \label{Willie_hypo_oh}
\begin{cases}
 H_{0} :& \mathbf{y}_{w}[i]= \mathbf{n}_{w}[i],\\
 H_{1} :& \mathbf{y}_{w}[i]=\sqrt{P_{a}H_{aw}}\mathbf{x}[i] +\mathbf{n}_{w}[i],
\end{cases}
\end{aligned}
\end{equation}
where $\mathbf{n}_{w}[i]$ is the AWGN with mean zero and variance $\sigma_{w}^{2}$ at Willie. 

\subsubsection{TH Transmission Scheme}
Similarly, to detect the transmission of the UAV to the relay under TH transmission scheme, Willie should distinguishes the following two received signals.
\begin{equation}
\begin{aligned} \label{Willie_hypo_th}
 \!\begin{cases}
H_{0} :& \mathbf{y}_{w}[i]=\sqrt{P_{r}H_{rw}}\mathbf{t}[i]+ \mathbf{n}_{w}[i],\\
H_{1} :& \mathbf{y}_{w}[i]=\sqrt{P_{a}H_{aw}}\mathbf{x}[i] +\mathbf{n}_{w}[i].
\end{cases}
\end{aligned}
\end{equation}

In the context of covert communication, Willie aims to minimize the probability of missed detection (MD) $P_{MD}$ and the probability of false alarm (FA) $P_{FA}$. Denote Willie's decisions in favor of $H_0$ and $H_1$ as $D_{0}$ and $D_{1}$, respectively. Thus, $P_{FA}$ can be expressed as $P_{FA}=Pr(D_{0}|H_{1})$, and $P_{MD}=Pr(D_{1}|H_{0})$. Therefore, the total DEP $\xi$ at Willie is formulated as 
\begin{equation}\label{total_error_rate}
\xi = \pi_{0}P_{FA}+\pi_{1}P_{MD}.
\end{equation}

\section{Minimum Covert Transmission Time Under the OH Transmission Scheme}\label{sec_tran_stra}

In this section, as the UAV adopts the OH (o) transmission scheme, for a given hover location $\mathbf{U}(\mathbf{q}_{a}, h)$ of the UAV, we first present the optimal DEP at Willie. Based on the optimal DEP, we formulate the minimize the covert transmission time as an optimization problem. Then, the optimization problem is solved by jointly designing the optimal transmit power $P_{a}$ and the prior transmission probability $\rho_{1}$. With the altitude $h$ of the UAV, the optimal horizontal location $\mathbf{q}_{a}$ of the UAV is non-convex and challenging to solve directly, we thus propose a PSO-based optimization algorithm to get the optimal horizontal location.

\subsection{Detection Error Probability at Willie}\label{dep_oh}

We first analyze the optimal DEP at Willie. Assume that Willie adopts the optimal statistical hypothesis test that minimizes $\xi$, as per \cite{bash2013limits}, the optimal DEP $\xi^{*}$ at Willie can be determined as
\begin{equation}\label{min_error_rate_v}
\begin{aligned}
\xi^{*}& \ge \min(\rho_{0},\rho_{1})-\max(\rho_{0},\rho_{1})\mathcal{V}_{T}(\mathbb{P}_{1}^{n},\mathbb{P}_{0}^{n}),
\end{aligned}
\end{equation}
where $\mathcal{V}_{T}(\mathbb{P}_{1}^{n},\mathbb{P}_{0}^{n})$ is the total variation between $\mathbb{P}_{1}^{n}$ and $\mathbb{P}_{0}^{n}$. $\mathbb{P}_{0}^{n}$ (resp. $\mathbb{P}_{1}^{n})$ is the distribution of the sequence $\mathbf{y}_{w}=\{\mathbf{y}_{w}[i]\}_{i=1}^{n}$ received in $n$ channel use when $\mathcal{H}_{0}$ (resp. $\mathcal{H}_{1}$) is true.

Due to the intractability of the expression $\mathcal{V}_{T}(\mathbb{P}_{1}^{n},\mathbb{P}_{0}^{n})$, similar to \cite{bash2013limits}, we can obtain an upper bound of it according to Pinsker’s inequality, which is determined as
\begin{equation}\label{min_error_rate_v_up}
\mathcal{V}_{T}(\mathbb{P}_{1}^{n},\mathbb{P}_{0}^{n}) \leq \sqrt{\frac{1}{2} \mathcal{D}(\mathbb{P}_{1}^{n},\mathbb{P}_{0}^{n})},
\end{equation}
where $\mathcal{D}(\mathbb{P}_{1}^{n},\mathbb{P}_{0}^{n})$ denotes the relative entropy (also known as Kullback-Leibler (KL) divergence) between $\mathbb{P}_{1}^{n}$ and $\mathbb{P}_{0}^{n}$. From the chain rule for relative entropy, we have
\begin{equation}\label{total_convert_con}
\begin{aligned}
\mathcal{D}(\mathbb{P}_{1}^{n}, \mathbb{P}_{0}^{n}) = n \mathcal{D}(\mathbb{P}_{1},&\mathbb{P}_{0}).
\end{aligned}
\end{equation}
According to (\ref{Willie_hypo_oh}), the KL divergence $\mathcal{D}(\mathbb{P}_{1},\mathbb{P}_{0})$ from $\mathbb{P}_{1}$ to $\mathbb{P}_{0}$ can be calculated as
\begin{equation}\label{min_error_rate_v_up_f}
\begin{aligned}
\mathcal{D}(\mathbb{P}_{1},\mathbb{P}_{0})=\frac{1}{2}\left [ \gamma_{aw}^{(o)}-\ln \left ( 1+\gamma_{aw}^{(o)}\right ) \right],
\end{aligned}
\end{equation}
where $\gamma_{aw}^{(o)}=\frac{P_{a}H_{aw}}{\sigma^{2}_{w}}$ denotes the received SNR of Willie under the OH transmission scheme, $P_{a}$ is the transmit power of UAV.

Overall, the lower bound of the optimal DEP $\xi^{*}$ at Willie can be expressed as
\begin{equation}\label{th_min_convert_con}
\begin{aligned}
\xi^{*} \ge\min(\rho_{0},\rho_{1})-\max(\rho_{0},\rho_{1})\sqrt{\frac{n}{4}\left(\gamma_{aw}^{(o)}-\ln \left ( 1+\gamma_{aw}^{(o)}\right )\right )}.
\end{aligned}
\end{equation}

\subsection{Transmission Time Minimization}\label{oh_time}
To achieve the minimum covert transmission time, for a given hover location $\mathbf{U}$ of the UAV, we first calculate the effective multicast throughput $C_{g}$ from the UAV to the $g$-th GU. With the non-negligible decoding error probability of finite channel use, according to \cite{xiang2020secure}, the effective throughput $C_{ag}$ can be given by
\begin{equation}\label{effect_throughput}
\begin{aligned}
C_{g}=\rho_{1}nR (1- \eta_{ag} ),
\end{aligned}
\end{equation}
where $R$ is the given target rate and $\eta_{ag} $ is decoding error probability at $g$-th GU. As per \cite{xiang2020secure}, $\eta_{ag}$ can be determined as
\begin{equation}\label{decode_error_pro}
\eta_{ag}  =  Q\left ( \frac{\sqrt{n}(1+\gamma _{ag})(\ln (1+\gamma _{ag})+\frac{1}{2}\ln n-R\ln 2) }{\sqrt{\gamma _{ag}(\gamma _{ag}+2)} } \right ),
\end{equation}
where $Q(\cdot)$ is the Q-function; $\gamma_{ag}$ denotes the received SNR and $\sigma_{g}^{2}$ is the power of AWGN at $g$-th GU, respectively. To make (\ref{decode_error_pro}) mathematically tractable for the minimization of transmission time, we employ a linear approximation function for the Q-function as detailed in \cite{durisi2016toward}. Thus, the decoding error $\eta_{ag}$ can be rewrite as
\begin{equation}\label{decoding_error_rate}
\begin{aligned}
\eta_{ag} \approx \begin{cases}
1,&   \gamma _{ag}< \vartheta-\frac{1}{2\varsigma} ,\\
-\varsigma  (\gamma _{ag}-\vartheta )+\frac{1}{2},&\vartheta-\frac{1 }{2\varsigma} \leq \gamma _{ag}\leq \vartheta+\frac{1 }{2\varsigma}, \\
 0,&  \gamma _{ag}>\vartheta+\frac{1}{2\varsigma }.
\end{cases}
\end{aligned}
\end{equation}
where $\varsigma =\sqrt{\frac{n}{2\pi (\exp(2R)-1)}} $ and $\vartheta=\exp (R)-1$.
We focus on the case where $0 < \eta_{ag}< 1$. By substituting this specific case of (\ref{decoding_error_rate}) into (\ref{effect_throughput}), the effective throughput $C_{g}$ from the UAV to $g$-th GU can be determined as
\begin{equation}\label{transmission_throughput_overall_app}
\begin{aligned}
C_{g} = \frac{1}{2} \rho_{1} n R \left(1+2\varsigma  (\gamma _{ag}-\vartheta )\right).
\end{aligned}
\end{equation}
The corresponding transmission time from UAV to the $g$-th GU is thus given by
\begin{equation}\label{transmission_time_overall_app}
\begin{aligned}
T_{g} = \frac{2M}{ \rho_{1} n R \left(1+2\varsigma  (\gamma _{ag}-\vartheta )\right)}.
\end{aligned}
\end{equation}
We then given the covertness requirement, as per \cite{lu2021joint, wu2023irs}, the fundamental covertness requirement can be expressed as
\begin{equation}\label{covert_constraint_oh}
\begin{aligned}
\xi^{*} &\ge \min(\rho_{0},\rho_{1})- 2\min(\rho_{0},\rho_{1})\epsilon,
\end{aligned}
\end{equation}
where $\epsilon>0$ is the covertness constraint. According to (\ref{th_min_convert_con}), the covertness requirement can be finally expressed as
\begin{equation}\label{covert_constraint_oh_final}
\begin{aligned}
\frac{\max(\rho_{0},\rho_{1})}{4\min(\rho_{0},\rho_{1})} \sqrt{n\left(\gamma_{aw}^{(o)}-\ln \left ( 1+\gamma_{aw}^{(o)}\right )\right )} \le \epsilon.
\end{aligned}
\end{equation}
Notice that the overall transmission time $T$ is fundamentally limited by the bottleneck link, which can be given by
\begin{equation}\label{overall_transmission_time}
\begin{aligned}
T= \mathop{\max}\limits_{g \in \mathbb{G} }\quad  T_{g}.
\end{aligned}
\end{equation}
Overall, given the hover location $\mathbf{U}$ of the UAV, the problem of minimizing the overall transmission time while ensuring the covertness can be formulated as
\begin{subequations} \label{optimal_time}
\begin{align}
\mathop{\min}\limits_{P_{a}, \rho_{1} } & \quad  T  \label{p_time},\\
 s.t. \quad  & \frac{\max(\rho_{0},\rho_{1})}{4\min(\rho_{0},\rho_{1})} \sqrt{n\left(\gamma_{aw}^{(o)}-\ln \left ( 1+\gamma_{aw}^{(o)}\right )\right )} \le \epsilon \label{p_covert}, \\
\quad  & P_{min}\le P_{a} \label{power_con}, \\
 \quad  & \rho_{min}<\rho_{1}<\rho_{max} \label{p_priori},
\end{align}
\end{subequations}
where $P_{min}=\frac{(\vartheta-\frac{1}{2\varsigma})\sigma _{\hat{g}}^{2}}{H_{a\hat{g}}}$ in (\ref{power_con}) ensures the feasibility of the specific case in (\ref{decoding_error_rate}), $\hat{g}$ denotes the GU with worst communication link from it to the UAV, $0<\rho_{min} \le 0.5$ and $0.5<\rho_{max} < 1$ to keep that the denominator in (\ref{p_covert}) will not be 0.
 

\subsection{Joint Optimal Transmit Power and Prior Probability with the Given Hover Location for the UAV}
Given the hover  location $\mathbf{U}$ of the UAV, the joint solution for the optimal transmit power and prior transmission probability, and the corresponding minimum transmission time can be summarized in the following theorem.

\begin{theorem}\label{th_min_time}
For the considered UAV multicast system, when the UAV transmits $M$ bits CI to the GUs positioned at the horizontal coordinate $\mathbf{q}_{a}$ and altitude $h$ over $n$ channel use in each time slot, along with the target rate $R$ and the covertness constraint $\epsilon$. The optimal transmit power $P_{a}^{*}$ and prior transmission probability $\rho_{1}^{*}$, and the associated minimum transmission time $T^{*}$ for the optimization problem (\ref{optimal_time}) are determined as
\begin{equation}\label{optimal_transmission_power}
\begin{aligned}
P_{a}^{*} = \frac{4\epsilon\sigma^{2}_{w}(1-\rho_{1}^{*})}{H_{aw}\rho_{1}^{*}}\sqrt{\frac{2}{n}}
\end{aligned}
\end{equation}
and
\begin{equation}\label{optimal_transmission_prior}
\begin{aligned}
\rho_{1}^{*} = \begin{cases}
0.5,&   A \ge \frac{1-2\varsigma \vartheta }{2\varsigma },\\
\rho_{max},& otherwise, 
\end{cases}
\end{aligned}
\end{equation}
where $\rho_{max} = \frac{4\epsilon\sigma^{2}_{w}H_{a\hat{g}}}{4\epsilon\sigma^{2}_{w}H_{a\hat{g}} + \left(\vartheta - \frac{1}{2\varsigma}\right)\sigma_{\hat{g}}^{2} H_{aw} \sqrt{\frac{n}{2}}}$ and $A=\frac{4\epsilon\sigma^{2}_{w}H_{a\hat{g}}}{\sigma^{2}_{g}H_{aw} }\sqrt{\frac{2}{n}}$; $\hat{g}$ represent the farthest GU from the UAV; $\varsigma$ and $\vartheta$ are given in (\ref{decoding_error_rate}); $H_{aw}$ and $H_{a\hat{g}}$ can be obtained by substituting $\mathbf{q}_{a}$ and $h$ into (\ref{channel_gain_ag}) and (\ref{channel_gain_aw}), respectively.
The corresponding minimum transmission time is thus determined as
\begin{equation}\label{optimal_transmission_time}
\begin{aligned}
T^{*}=\frac{2M}{\rho_{1}^{*}n R \left(1+ 2\varsigma  (\gamma _{a\hat{g}}^{*}-\vartheta ) \right)},\\
\end{aligned}
\end{equation}
where $\gamma_{a\hat{g}}^{*}=\frac{P_{a}^{*}H_{a\hat{g}}}{\sigma^{2}_{\hat{g}}}$.
\end{theorem}
\begin{proof}
For a given location $\mathbf{U}$ of the UAV, according to (\ref{channel_gain_ag}), we can conclude that the worst communication link from the UAV to GUs is the one farthest from the UAV to GU, using $\hat{g}$ denotes the GU farthest from the UAV, the overall transmission time refers to the transmission time from the UAV to GU $\hat{g}$ (i.e., $T=T_{\hat{g}}$). Then, based on this worst communication link, we jointly explore the optimal transmit power and prior transmission probability for the UAV to minimize the transmission time. From (\ref{transmission_time_overall_app}), it is easy to find that $T_{\hat{g}}$ is a decreasing function of $C_{\hat{g}}$. The optimal problem can be transformed into maximizing $C_{\hat{g}}$. We first calculate the optimal transmit power while satisfying the covertness constraint. According to (\ref{p_covert}), the covertness requirement can be rewritten as

\begin{equation}\label{min_convert_con}
\begin{aligned}
\frac{P_{a} H_{aw}}{\sigma^{2}_{w}}-\ln \left ( 1+\frac{P_{a} H_{aw}}{\sigma^{2}_{w}}\right )  \leq \frac{16\epsilon^{2}}{n}\left(\frac{\min(\rho_{0},\rho_{1})}{\max(\rho_{0},\rho_{1})}\right)^{2}.
\end{aligned}
\end{equation}

Given the typically low transmit power in covert communications and thus a low SNR, we can leverage the Taylor series expansion and the inequality $\ln(1+x)\ge x-\frac{x^{2}}{2}$. Thus, (\ref{min_convert_con}) can be further approximated as

\begin{equation}\label{min_convert_con_app}
\begin{aligned}
\frac{P_{a} H_{aw}}{\sigma^{2}_{w}} \leq 4\epsilon\sqrt{\frac{2}{n}}\frac{\min(\rho_{0},\rho_{1})}{\max(\rho_{0},\rho_{1})}.
\end{aligned}
\end{equation}
The optimal transmit power $P_{a}^{*}$ is attained by setting the inequality in (\ref{min_convert_con_app}) as equality, which can be given by
\begin{equation}\label{op_power_maxmin_oh}
\begin{aligned}
P_{a}^{*}=\frac{4\epsilon\sigma^{2}_{w}}{H_{aw}}\sqrt{\frac{2}{n}}\frac{\min(\rho_{0},\rho_{1})}{\max(\rho_{0},\rho_{1})}.
\end{aligned}
\end{equation}
From (\ref{op_power_maxmin_oh}), the optimal transmit power is highly related to the prior probabilities. Therefore, we address the optimization problem by considering two cases: $\rho_{1} \le \rho_{0}$ and $\rho_{1}> \rho_{0}$.

\subsubsection{ $\rho_{1} \le \rho_{0}$}
With $\rho_{1} \le \rho_{0}$, the optimal transmit power can be determined as 
\begin{equation}\label{op_power_case1_oh}
\begin{aligned}
P_{a}^{*}=\frac{4\epsilon\sigma^{2}_{w}\rho_{1}}{H_{aw}(1-\rho_{1})}\sqrt{\frac{2}{n}}.
\end{aligned}
\end{equation}
Substituting $P_{a}^{*}$ into (\ref{transmission_throughput_overall_app}), according to the monotonicity of the product of increasing functions, $\frac{\partial C_{\hat{g}}}{\partial \rho_{1}} >0$. Thus, we can conclude that, for $\rho_{1}>0$, $C_{\hat{g}}$ monotonically increases with $\rho_{1}$, the optimal $\rho_{1}^{*}$ is the maximum value of its possible range. Hence, the optimal $\rho_{1}^{*}=0.5$. The corresponding optimal transmit power $P_{a}^{*}=\frac{4\epsilon\sigma^{2}_{w}}{H_{aw}}\sqrt{\frac{2}{n}}$.

\subsubsection{ $\rho_{1}\ge \rho_{0}$} Similarly, with $\rho_{1}\ge \rho_{0}$, the optimal transmit power can be given by
\begin{equation}\label{op_power_case2_oh}
\begin{aligned}
P_{a}^{*}=\frac{4\epsilon\sigma^{2}_{w}(1-\rho_{1})}{H_{aw}\rho_{1}}\sqrt{\frac{2}{n}}.
\end{aligned}
\end{equation}
Substituting $P_{a}^{*}$ into (\ref{transmission_throughput_overall_app}), the effective throughput can be expressed as
\begin{equation}\label{op_ec_fun}
\begin{aligned}
C_{\hat{g}}=\rho_{1} n R \left(1+2\varsigma  \left( \frac{A}{\rho_{1}}- A -\vartheta \right)\right),
\end{aligned}
\end{equation}
where $A=\frac{4\epsilon\sigma^{2}_{w}H_{a\hat{g}}}{\sigma^{2}_{g}H_{aw} }\sqrt{\frac{2}{n}}$. The derivative of $C_{\hat{g}}$ can be given by
\begin{equation}\label{op_ec_fun}
\begin{aligned}
\frac{\partial C_{\hat{g}}}{\partial \rho_{1}} =n R \left(1-2\varsigma  \left(  A +\vartheta \right)\right).
\end{aligned}
\end{equation}
Thus, if $ A \ge \frac{1-2\varsigma \vartheta }{2\varsigma } $, the optimal $\rho_{1}^{*}$ is the minimum value of its possible range, which is $\rho_{1}^{*}=0.5$. The corresponding optimal transmit power $P_{a}^{*}=\frac{4\epsilon\sigma^{2}_{w}}{H_{aw}}\sqrt{\frac{2}{n}}$. If $ A < \frac{1-2\varsigma \vartheta }{2\varsigma } $, the optimal $\rho_{1}^{*}$ is the maximum value of its possible range. According to the transmit power constraint of (\ref{power_con}), the upper bound of the prior transmission probability can be given by 
\begin{equation}\label{op_upper_bound_rho}
\begin{aligned}
\rho_{max} = \frac{4\epsilon\sigma^{2}_{w}H_{a\hat{g}}}{4\epsilon\sigma^{2}_{w}H_{a\hat{g}} + \left(\vartheta - \frac{1}{2\varsigma}\right)\sigma_{\hat{g}}^{2} H_{aw} \sqrt{\frac{n}{2}}}.
\end{aligned}
\end{equation}
Thus, the the optimal $\rho_{1}^{*}= \rho_{max}$. The corresponding optimal transmit power is by substituting  $\rho_{1}^{*}$ into (\ref{op_power_case2_oh}).

Overall, the jointly optimal transmit power and prior transmission probability can be given by (\ref{optimal_transmission_power}) and (\ref{optimal_transmission_prior}), the corresponding minimum transmission time is determined as (\ref{optimal_transmission_time}). 
\end{proof}

\subsection{Horizontal Location Optimization for the UAV}
In this subsection, we intend to jointly explore the joint optimal transmit power, prior transmission probability and horizontal location of the UAV to achieve the overall minimum transmission time. The optimal problem is challenging to solve mathematically since the variability in the worst communication links is associated with different horizontal locations. We thus adopt the PSO algorithm \cite{sanchez2019distributed} to find the optimal horizontal location of the UAV. In the PSO algorithm, a swarm of particles moves in a \textit{D}-dimensional search space to optimize the fitness function. During the searching process, each particle assesses its fitness for every iteration and adjusts its position based on its historical and global best positions. The $m$-th particle's position is determined as
\begin{equation} \label{pso_o}
\begin{aligned}
\mathbf{o}_{m}(k+1)=\mathbf{o}_{m}(k)+\mathbf{v}_{m}(k+1).
\end{aligned}
\end{equation}
where $k$ is the iteration index; $\mathbf{o}_{m}(k)$ is current position and $\mathbf{v}_{m}(k+1)$ is the velocity of the $m$-th particle, which is given by
\begin{equation} \label{pso_v}
\begin{aligned}
 \mathbf{v}_{m}(k+1)=& w\mathbf{v}_{m}(k)+c_{1}\psi_{1}(\mathbf{P}_{m}-\mathbf{o}_{m}(k))\\&+c_{2}\psi_{2}(\mathbf{P}_{g}-\mathbf{o}_{m}(k)),
\end{aligned}
\end{equation}
where $w$ is the inertial weight coefficient of each particle; $c_1$ and $c_2$ are acceleration constants; $\psi_{1}$ and $\psi_{2}$ are the local and global learning coefficients; $\mathbf{P}_{m}$ and $\mathbf{P}_{g}$ represent the local best and global best particle's position.

In this paper, the fitness function refers to the minimum transmission time, each particle's position is bounded by the concerned area, we need to check if the new position satisfies this constraint. Algorithm~\ref{algo_pso} illustrates the details of the PSO-based optimization algorithm that jointly solves the transmit power and hovering location of UAV.

\begin{algorithm} \label{algo_pso}
	\SetAlgoLined
	\SetKw{return}{return}
	  \caption{PSO-based location optimization algorithm}
	  \KwData{Location of GUs and Willie; Target rate R; Channel use n; Noise variance of GUs and Willie $\sigma_g$ and $\sigma_w$; Altitude of the UAV $h$; CI size $M$; Channel fading loss exponent $\alpha$; Covertness constraint $\epsilon$; Number of particles $m_{max}$; Max iterations $k_{max}$;}
	  \KwResult{Optimal transmit power $P_{a}^{*}$, optimal prior transmission probability $\rho_{1}^{*}$; horizontal location $\mathbf{q}_{a}^{*}$; Minimum transmission time $T^{*}$;}

  Initialize random positions and velocities for $m_{max}$ particles\;

  \For{$k= 1; k < k_{max}; k++$}{

    \For{$m= 1; m < m_{max}; m++$}{
    Update $m$-th particle's position by (\ref{pso_o})\;
      Check and correct the particle's position in the concerned area\;
      Calculate the optimal transmit power $P_{a}$, prior transmission probability $\rho_{1}$ and minimum transmission time $T_{m}(k)$ according to Theorem~\ref{th_min_time} with the current particle's position $\mathbf{o}_{m}(k)$ \;

      \If{$T_{m}(k)$ is less than personal minimum transmission time $T_{m}^{*}$}{
      	 $T_{m}^{*}=T_{m}(k)$\;
        $\mathbf{P}_{m}=\mathbf{o}_{m}(k)$\;
        \If{$T_{m}^{*}$ is less than global minimum transmission time $T^{*}$}{
        $P_{a}^{*}=P_{a}$\;
        $\rho_{1}^{*}=\rho_{1}$\;
        $T^{*}=T_{m}^{*}$\;
        $\mathbf{q}_{a}^{*}=\mathbf{P}_{g}=\mathbf{P}_{m}$\;
        }
      }
    }
  }
\end{algorithm}

\section{Minimum Covert Transmission Time Under the TH Transmission Scheme}\label{sec_tran_stra_th}
In this section, the UAV adopts the TH (t) transmission scheme to minimize the completion transmission time. With a selected GU as the relay, we first present the optimal DEP at Willie. Leveraging this optimal DEP, we then formulate the minimization of covert transmission time as an optimization problem. We solve the optimization problem by jointly designing the optimal transmit power $P_{a}$ and prior transmission probability $\rho_{1}$. The selection of the optimal relay is non-convex and challenging to solve mathematically, we thus propose a exhaustive optimization algorithm to get the optimal relay.

\subsection{Detection Error Probability at Willie}

We first analyze the optimal DEP at Willie. Similar to Section \ref{dep_oh}, suppose Willie adopts the optimal statistical hypothesis test that minimizes $\xi$ the optimal DEP $\xi^{*}$ at Willie, therefore, the lower bound of the optimal can be determined as
\begin{equation}\label{min_error_rate_v}
\begin{aligned}
\xi^{*} \ge \min(\rho_{0},\rho_{1})-\max(\rho_{0},\rho_{1})\sqrt{\frac{n\mathcal{D}(\mathbb{P}_{1},\mathbb{P}_{0})}{2}}.
\end{aligned}
\end{equation}

The KL divergence $\mathcal{D}(\mathbb{P}_{1},\mathbb{P}_{0})$ from $\mathbb{P}_{1}$ to $\mathbb{P}_{0}$ can be calculated as
\begin{equation}\label{min_error_rate_v_up_f}
\begin{aligned}
\mathcal{D}(\mathbb{P}_{1},\mathbb{P}_{0})&=\frac{1}{2}\left [\gamma_{aw}^{(h)}-\ln \left ( 1+\gamma_{aw}^{(h)}\right ) \right],
\end{aligned}
\end{equation}
where $\gamma_{aw}^{(h)}=\frac{P_{a}H_{aw}}{P_{r}H_{rw}+\sigma^{2}_{\hat{w}}}$, $P_{a}$ and $P_{r}$ are the transmit power of the UAV and the relay, respectively.

Thus, the lower bound of the optimal DEP $\xi^{*}$ at Willie can be expressed as
\begin{equation}\label{th_min_convert_con_th}
\begin{aligned}
\xi^{*} \ge \min(\rho_{0},\rho_{1})-\max(\rho_{0},\rho_{1}) \sqrt{\frac{n}{4}\left(\gamma_{aw}^{(h)}-\ln \left ( 1+\gamma_{aw}^{(h)}\right )\right )}.
\end{aligned}
\end{equation}

\subsection{Transmission Time Minimization}
Similar to Section~\ref{oh_time}, with the selected relay, to achieve the minimum covert transmission time, we need to calculate the effective throughput $C_{g}$ of the concerned system, which depends on the decoding error probability. According to the SNR received at the GUs, the probability of decoding the block error at the GUs just relying on the direct links from the UAV to the GUs is higher than the MRC-based relay links. Thus, our focus is the MRC-based relay links. According to \cite{hu2015capacity}, the overall decoding correct probability is the intersection of the decoding correct probability at the relay and the decoding correct probability at the $g$-th GU with MRC. Therefore, the overall decoding error probability, which represents the likelihood of at least one decoding error occurring in the system, can be given by 
\begin{equation}\label{effect_throughput_th_res}
\begin{aligned}
\eta_{ag}=1-(1-\eta_{ar})(1-\eta_{arg}),
\end{aligned}
\end{equation}
where $\eta_{ar}$ and $\eta_{arg}$ are the decoding error probabilities at the selected relay and $g$-th GU with MRC, respectively. Besides, due to the UAV transmits one block with probability $\rho_{1}$, the relay transmits one block with probability $\rho_{0}$, hence, the overall probability that to transmit one block is $\min(\rho_{0}, \rho_{1})$. Therefore, the effective throughput $C_{g}$ can be expressed as
\begin{equation}\label{effect_throughput_th_final}
\begin{aligned}
C_{g}=\min(\rho_{0},\rho_{1})nR(1- \eta_{ar} )(1- \eta_{arg} ).
\end{aligned}
\end{equation}
Similarly, we employ the linear approximation function for the decoding error probability. $\eta_{ar}$ at the selected relay can be determined by
\begin{equation}\label{decoding_error_rate_th_r_math}
\begin{aligned}
\eta_{ar}\approx \begin{cases}
1,&   \gamma _{ar}< \vartheta-\frac{1}{2\varsigma} ,\\
-\varsigma  (\gamma _{ar}-\vartheta )+\frac{1}{2},&\vartheta-\frac{1 }{2\varsigma} \leq \gamma _{ar}\leq \vartheta+\frac{1 }{2\varsigma}, \\
 0,&  \gamma _{ar}>\vartheta+\frac{1}{2\varsigma },
\end{cases}
\end{aligned}
\end{equation}
and $\eta_{arg}$ can be expressed as
\begin{equation}\label{decoding_error_rate_th_g_math}
\begin{aligned}
\eta_{arg}\approx \begin{cases}
1,&   \gamma _{arg}< \vartheta-\frac{1}{2\varsigma} ,\\
-\varsigma  (\gamma _{arg}-\vartheta )+\frac{1}{2},&\vartheta-\frac{1 }{2\varsigma} \leq \gamma _{arg}\leq \vartheta+\frac{1 }{2\varsigma}, \\
 0,&  \gamma _{arg}>\vartheta+\frac{1}{2\varsigma }.
\end{cases}
\end{aligned}
\end{equation}
where $\gamma_{ar}$ and $\gamma_{arg}$ are defined in Section~\ref{transmission_schemes}, $\varsigma$ and $\vartheta$ are given in (\ref{decoding_error_rate}).
We focus on the cases where $0 < \eta_{ar} < 1$ and $0 < \eta_{arg} < 1$. By substituting this specific cases of (\ref{decoding_error_rate_th_r_math}) and (\ref{decoding_error_rate_th_g_math}) into (\ref{effect_throughput_th_final}), the effective throughput $C_{g}$ of $g$-th GU can be given by
\begin{equation}\label{throughput_overall_th}
\begin{aligned}
C_{g} = \frac{1}{4}\min(\rho_{0},\rho_{1}) n R (1+2\varsigma  (\gamma _{ar}-\vartheta ))(1+2\varsigma  (\gamma _{arg}-\vartheta )).
\end{aligned}
\end{equation}
The corresponding transmission time is then determined by
\begin{equation}\label{transmission_time_overall_th}
\begin{aligned}
T_{g} =\frac{4M}{\min(\rho_{0},\rho_{1}) n R (1+2\varsigma  (\gamma _{ar}-\vartheta ))(1+2\varsigma  (\gamma _{arg}-\vartheta ))}.
\end{aligned}
\end{equation}
The overall transmission time $T$ is same as (\ref{overall_transmission_time}). Then, based on (\ref{th_min_convert_con_th}) and (\ref{transmission_time_overall_th}), the problem of minimizing the overall transmission time while ensuring the covertness can be formulated as
\begin{subequations} \label{optimal_time_th}
\begin{align}
\mathop{\min}\limits_{P_{a}, \rho_{1}} & \quad  T\label{p_time_th},\\
 s.t. \quad  & \frac{\max(\rho_{0},\rho_{1})}{4\min(\rho_{0},\rho_{1})} \sqrt{n\left(\gamma_{aw}^{(h)}-\ln \left ( 1+\gamma_{aw}^{(h)}\right )\right )} \le \epsilon \label{p_covert_th}, \\
\quad  & P_{min}\le P_{a} \label{power_con_th}, \\
 \quad  & \rho_{min}<\rho_{1}<\rho_{max} \label{p_priori_th},
\end{align}
\end{subequations}
where $P_{min}=\max\left(\frac{(\vartheta-\frac{1}{2\varsigma})\sigma _{r}^{2}}{H_{ar}},\frac{(\vartheta-\frac{1}{2\varsigma})\sigma _{\hat{g}}^{2}-P_{r}H_{r\hat{g}}}{H_{a\hat{g}}}\right)$ to ensure the feasibility of the specific case in (\ref{decoding_error_rate_th_r_math}) and (\ref{decoding_error_rate_th_g_math}).

\subsection{Joint Optimal Transmit Power and Prior Probability with the Selected Relay}
Given the selected relay, the joint solution of the optimal transmit power and prior transmission probability, and the corresponding minimum transmission time can be summarized in the following theorem.
\begin{theorem}\label{th_min_time_th}
For the considered UAV multicast system, when the UAV transmits $M$ bits CI to the GUs positioned above of the selected relay $r$ and altitude $h$ over $n$ channel use in each time slot, along with the target rate $R$ and the covertness constraint $\epsilon$. The expected optimal transmit power $P_{a}^{*}$ and the prior transmission probability $\rho_{1}^{*}$, and the associated minimum transmission time $T^{*}$ for the optimization problem (\ref{optimal_time_th}) are determined as
\begin{equation}\label{optimal_transmission_power_th}
\begin{aligned}
P_{a}^{*}=\frac{4\epsilon \sigma^{2}_{w}}{ H_{aw}e^{\kappa _{rw} }\kappa _{rw}E_{1}(\kappa _{rw})}\sqrt{\frac{2}{n}},
\end{aligned}
\end{equation}
and
\begin{equation}\label{optimal_transmission_prior_th}
\begin{aligned}
\rho_{1}^{*} = 0.5, 
\end{aligned}
\end{equation}
where $E_{1}(\cdot)$ is the exponential integral, defined as $E_{1}(z)=\int_{z}^{\infty } e^{-t}t^{-1}dt$ \cite{gradshteyn2014table}, and $\kappa _{rw}=\frac{\sigma^{2}_{w}}{P_{r}l_{rw}^{\alpha _{N}}}$. The corresponding expected minimum transmission time is thus determined as
\begin{equation}\label{optimal_transmission_time_th}
\begin{aligned}
T^{*}=\frac{8M}{  n R (1+2\varsigma  (\gamma _{ar}^{*}-\vartheta ))(1+2\varsigma  (\gamma _{arg}^{*}-\vartheta ))},\\
\end{aligned}
\end{equation}
where $\gamma_{ar}^{*}=\frac{P_{a}^{*}H_{ar}}{\sigma^{2}_{r}}$ and $\bar{\gamma}_{ar\hat{g}}^{*}=\frac{P_{a}^{*}H_{a\hat{g}}+P_{r}l_{r\hat{g}}^{\alpha _{N}}}{\sigma^{2}_{\hat{g}}}$; $\hat{g}$ represent the farthest GU from the relay; $\varsigma$ and $\vartheta$ are given in (\ref{decoding_error_rate}).
\end{theorem}
\begin{proof}
According to (\ref{channel_gain_ag}) and (\ref{channel_gain_rg}), given the selected relay, we can conclude that the worst communication link from the UAV-relay-GUs is the one farthest from the relay to GU, using $\hat{g}$ denotes the GU farthest from the selected relay. Then, based on this worst communication link, we jointly explore the optimal transmit power and prior transmission probability for the UAV to minimize the transmission time. From (\ref{transmission_time_overall_th}), it is easy to find that $T_{\hat{g}}$ is a decreasing function of $C_{\hat{g}}$. The optimal problem can be transformed into maximizing $C_{\hat{g}}$. Notice that both the UAV and relay only know the channel state information from GUs, we first calculate the expected throughput $\bar{C}_{\hat{g}}$, which can be given by 
\begin{align}\label{expected_throughput_th}
\bar{C}_{\hat{g}}&= \int_{0}^{\infty } \!B \left(1+2\varsigma  \left(\frac{P_{a}H_{ag}+P_{r}l_{rg}^{\alpha_{N}}\vert h_{rg}\vert ^{2}}{\sigma^{2}_{g}}\!-\!\vartheta \right)\right) f_{\vert h_{rg}\vert ^{2}}(x)dx \notag \\
&=B \left(1+2\varsigma  \left(\frac{P_{a}H_{ag}+P_{r}l_{rg}^{\alpha_{N}}}{\sigma^{2}_{g}}\!-\!\vartheta \right)\right) ,
\end{align}
where $B=\frac{1}{4}\min(\rho_{0},\rho_{1}) n R (1+2\varsigma  (\gamma _{ar}-\vartheta ))$. We can see $\bar{C}_{\hat{g}}$ is an increasing function with respect to $P_{a}$. To achieve the maximum $\bar{C}_{\hat{g}}$, we need to obtain the maximum transmit power while satisfying the covertness constraint. Therefore, similar to the approaches of (\ref{min_convert_con_app}), the covertness requirement can be rewritten as
\begin{equation}\label{min_convert_con_th}
\begin{aligned}
\mathbb{E}\left[\frac{P_{a}H_{aw}}{P_{r}H_{rw}+\sigma^{2}_{w}}\right] \leq 4\epsilon\sqrt{\frac{2}{n}}\frac{\min(\rho_{0},\rho_{1})}{\max(\rho_{0},\rho_{1})}.
\end{aligned}
\end{equation}
The expected value of KL divergence can be given by
\begin{equation}\label{min_convert_con_app_th}
\begin{aligned}
\mathbb{E}\left[\frac{P_{a}H_{aw}}{P_{r}H_{rw}+\sigma^{2}_{w}}\right]
&= \int_{0}^{\infty } \frac{P_{a}H_{aw}}{P_{r}l_{rw}^{\alpha_{N}}\vert h_{rw}\vert ^{2}+\sigma^{2}_{w}}f_{\vert h_{rw}\vert ^{2}}(x)dx\\
&=\frac{P_{a}H_{aw}e^{\kappa _{rw} }\kappa _{rw}}{\sigma^{2}_{w} }E_{1}(\kappa _{rw}),
\end{aligned}
\end{equation}
By setting the inequality of (\ref{min_convert_con_th}) as equality, the optimal transmit power $P_{a}^{*}$ can be given by
\begin{equation}\label{op_power_maxmin_th}
\begin{aligned}
P_{a}^{*}=\frac{4\epsilon \sigma^{2}_{w}}{ H_{aw}e^{\kappa _{rw} }\kappa _{rw}E_{1}(\kappa _{rw})}\sqrt{\frac{2}{n}}\frac{\min(\rho_{0},\rho_{1})}{\max(\rho_{0},\rho_{1})}.
\end{aligned}
\end{equation}
Then, similar to the proof of Theorem~\ref{th_min_time}, we address the optimization problem by considering two cases: $\rho_{1} \le \rho_{0}$ and $\rho_{1}> \rho_{0}$. 

\subsubsection{ $\rho_{1} \le \rho_{0}$} With $\rho_{1}\le \rho_{0}$, the optimal transmit power can be given by
\begin{equation}\label{op_power_case1_th}
\begin{aligned}
P_{a}^{*}=\frac{4\epsilon\sigma^{2}_{w}\rho_{1}}{H_{aw}e^{\kappa _{rw} }\kappa _{rw}E_{1}(\kappa _{rw})(1-\rho_{1})}\sqrt{\frac{2}{n}}.
\end{aligned}
\end{equation}
Substituting $P_{a}^{*}$ into (\ref{expected_throughput_th}), according to the  monotonicity of the product of increasing functions, $\frac{\partial \bar{C}_{\hat{g}}}{\partial \rho_{1}} >0$. We can conclude that, for $\rho_{1}>0$, $\bar{C}_{\hat{g}}$ monotonically increases with $\rho_{1}$, the optimal $\rho_{1}^{*}$ is the maximum value of its possible range. Hence, the optimal $\rho_{1}^{*}=0.5$. 

\subsubsection{ $\rho_{1} \ge \rho_{0}$} With $\rho_{1} \ge \rho_{0}$, the optimal transmit power can be given by
\begin{equation}\label{op_power_case1_th}
\begin{aligned}
P_{a}^{*}=\frac{4\epsilon\sigma^{2}_{w}(1-\rho_{1})}{H_{aw}e^{\kappa _{rw} }\kappa _{rw}E_{1}(\kappa _{rw})\rho_{1}}\sqrt{\frac{2}{n}}.
\end{aligned}
\end{equation}
Substituting $P_{a}^{*}$ into (\ref{expected_throughput_th}), according to the  monotonicity of the product of decreasing functions, $\frac{\partial \bar{C}_{\hat{g}}}{\partial \rho_{1}} <0$. We can conclude that, for $\rho_{1}>0$, $\bar{C}_{\hat{g}}$ monotonically decreases with $\rho_{1}$, the optimal $\rho_{1}^{*}$ is the minimum value of its possible range. Hence, the optimal $\rho_{1}^{*}=0.5$. 

Overall, the jointly optimal transmit power and prior transmission probability can be given by (\ref{optimal_transmission_power_th}) and (\ref{optimal_transmission_prior_th}), the corresponding minimum transmission time is determined as (\ref{optimal_transmission_time_th}). 
\end{proof}

\subsection{Relay Selection Optimization}
In this subsection, we intend to jointly explore the optimal transmit power and prior transmission probability of the UAV and optimal relay to achieve the overall minimum transmission time. The optimal problem is challenging to solve mathematically since the variability in the worst communication links is associated with different relay. We thus adopt the exhaustive algorithm to find the optimal relay, as shown in Algorithm~\ref{algo_relay}.

\begin{algorithm} \label{algo_relay}
	\SetAlgoLined
	\SetKw{return}{return}
	  \caption{Optimal relay selection algorithm}
	  \KwData{Location of GUs and Willie; Target rate R; Transmit power of the relay $P_{r}$; Channel use n; Noise variance of GUs and Willie $\sigma_g$ and $\sigma_w$; Altitude of the UAV $h$; CI size $M$; Channel fading loss exponent $\alpha^{L}$ and $\alpha^{N}$; Covertness constraint $\epsilon$;}
	  \KwResult{Optimal transmit power $P_{a}^{*}$; Optimal prior transmission probability $\rho_{1}^{*}$; Optimal relay $r^{*}$; Minimum transmission time $T^{*}$;}

  \For{$r= 1; r < G; r++$}{ 
    Find the GU $\hat{g}$ with the worst communication link\;
      Calculate the optimal transmit power $P_{a}(r)$, optimal prior transmission probability $\rho_{1}(r)$ and minimum transmission time $T(r)$ according to Theorem~\ref{th_min_time_th} at GU $\hat{g}$ with r as relay\; 
   \If{$T(r)$ is less than $T$}{
      	 $T^{*}=T(r)$\;
        $P_{a}^{*}=P_{a}(r)$\;
        $\rho_{1}^{*}=\rho_{1}(r)$\;
        $r^{*}=r(r)$\;
      }
  }
\end{algorithm}

\section{Numerical Results} \label{sec_num}
In this section, we present extensive numerical results to exhibit the performance of the proposed two covert transmission protocols.

\subsection{Simulation Settings}
In the simulations, we define a circular area with a radius of $500m$, and let the density parameter $\lambda$ varies in the range of $[2 \times 10^{-5} \pi^{-1}, 8 \times 10^{-4} \pi^{-1}]$. Aligning with the approach in \cite{shu2019delay}, we set the time slot duration as $\Delta t=1$ms, and each channel use lasts for 0.01ms. Consequently, the total number of channel uses in a time slot is $n=\frac{\Delta t}{0.01ms}=100$. Additional simulation parameters are as follows: UAV's altitude $h=500m$, file size $M=1$Mb, channel fading loss exponent $\alpha_{L}=-2$ and $\alpha_{N}=-3$, target rate $R=0.1$bpcu (bit per channel use), covertness constraint $\epsilon=0.1$, the S-curve parameters $e=4.88$ and $f=0.429$. We further set the noise variance for Willie $\sigma_{w}^{2}$ and noise variance $\sigma_{g}^{2}$ for each GU as $\sigma_{g}^{2}=\sigma_{w}^{2}=-70$dBm. For performance comparison, we utilize the center of the minimum circle covering all GUs with the fixed altitude $h=500m$ as the reference optimal location.

\subsection{Performance Evaluation}
This subsection provides numerical results to illustrate the impact of system parameters on the transmission time performance. For a given setting of system parameters (like the number of GUs, GU density, radius of considered area, and covertness constraint), we generate $10^{4}$ sets of GU spatial distributions, and then execute the Algorithm~\ref{algo_pso} and Algorithm~\ref{algo_relay} to get average minimum transmission time of UAV under the two transmission protocols.
\begin{figure} [thb]
    \centering
	\includegraphics[width=0.45\textwidth]{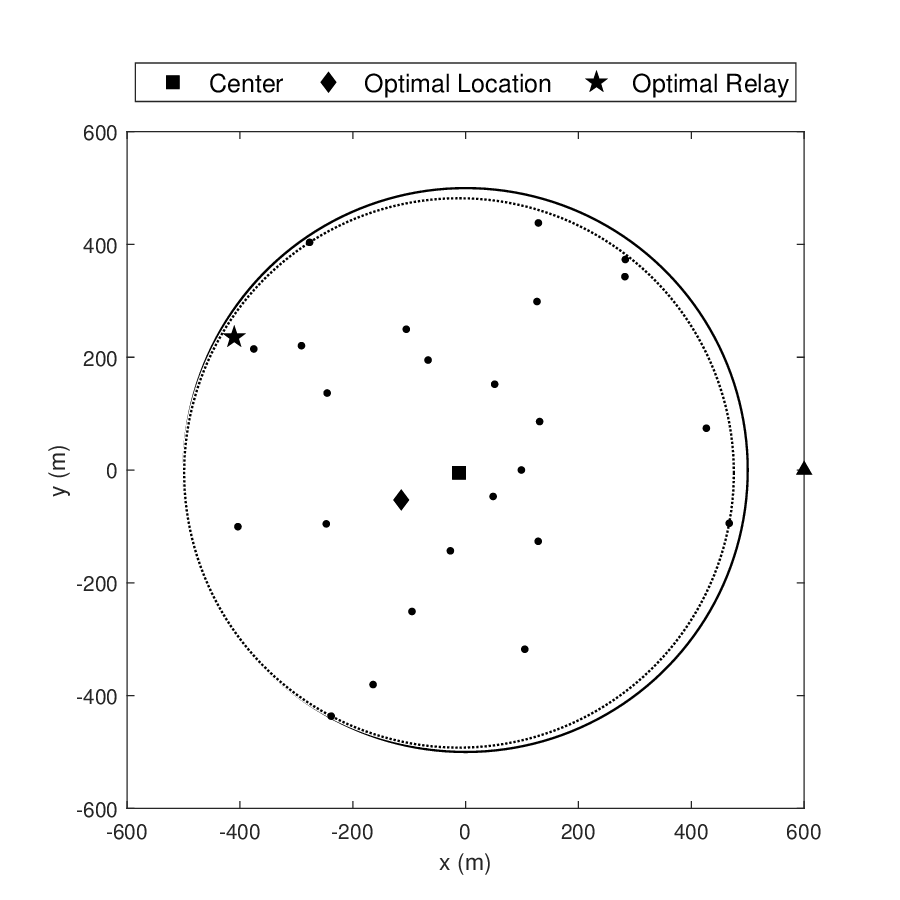}
    \caption{Illustration of the optimal location and relay under the settings of $\lambda=2*10^{-5}$, $r=500m$ and $\mathbf{q}_{w}=[600, 0]$.}
    \label{fig_sample}
\end{figure}

We provide an illustrative example of two transmission protocols in Fig. ~\ref{fig_sample}. The black dot represents GUs, the black triangle represents Willie, the dotted circle is the minimal circle covered all GUs, and the realization circle is the considered area. Surprisingly, the optimal UAV location of the OH protocol in the example does not align with and is even a little far away from the center of the minimum circle encompassing all GUs; the optimal relay of the TH protocol is even farther.

\begin{figure} [tb]
    \centering
	\includegraphics[width=0.45\textwidth]{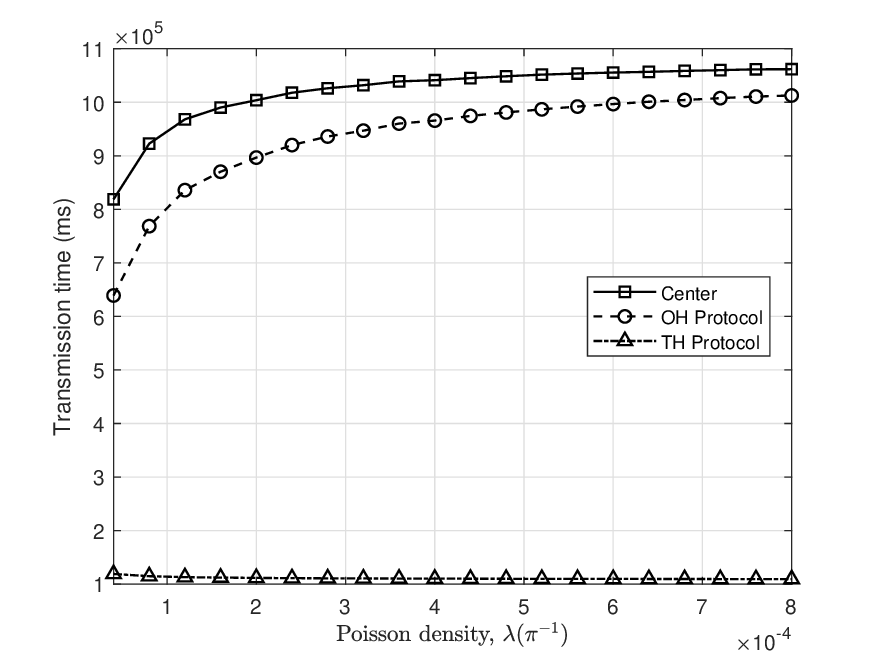}
    \caption{Transmission time versus $\lambda$ under the setting of $\mathbf{q}_{w}=[600, 0]$.}
    \label{fig_lambda_time}
\end{figure}

The results in Fig.~\ref{fig_lambda_time} highlight that both the OH and TH transmission protocols achieve shorter transmission time than simply taking the center of the minimum circle with the fixed altitude $h=500m$ as the location of UAV, and the TH transmission protocol significantly achieves shorter transmission time.  Underscoring efficiency of the new transmission protocols. Notably, a smaller value of $\lambda$ leads to a shorter transmission time. This observation suggests that the proposed transmission protocols are efficient in reducing the transmission time, particularly when the number of GUs is relatively small.

\begin{figure} [tb]
    \centering
	\includegraphics[width=0.45\textwidth]{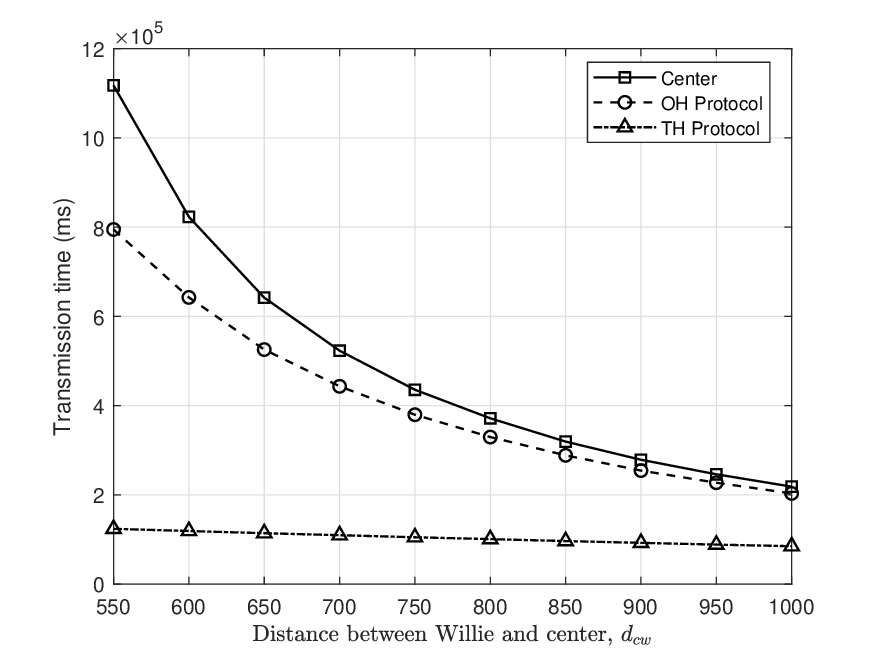}
    \caption{Transmission time versus $d_{cw}$ under the setting of $\lambda=4*10^{-5}*\pi^{-1}$.}
    \label{fig_distance_time}
\end{figure}

We then plot in Fig.~\ref{fig_distance_time} the relationship between transmission time and the distance $d_{cw}$ from Willie to the center of the considered circle with $\lambda=4*10^{-5}*\pi^{-1}$.  As can be seen from Fig.~\ref{fig_distance_time}, transmission time consistently diminishes as $d_{cw}$ increases. This behavior is attributed to the fact that, with Willie is farther away from GUs, the UAV can employ higher transmit power while satisfying the covertness constraint, resulting in a reduction in transmission time.

\begin{figure} [tb]
    \centering
	\includegraphics[width=0.45\textwidth]{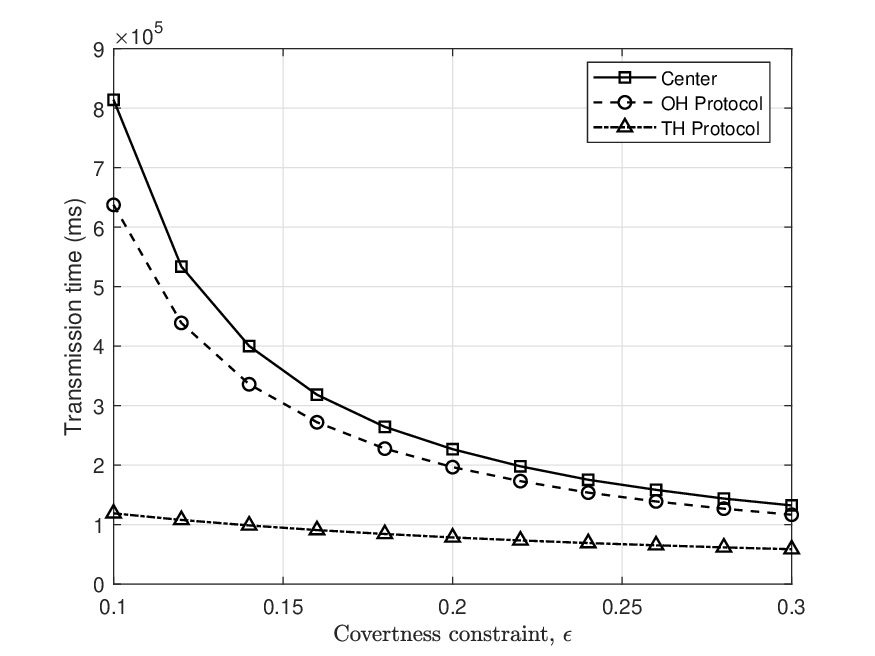}
    \caption{Transmission time versus covertness constraint $\epsilon$ under the settings of $\lambda=4*10^{-5}*\pi^{-1}$ and $\mathbf{q}_{w}=[600, 0]$.}
    \label{fig_covertness_time}
\end{figure}

Finally, we present in Fig.~\ref{fig_covertness_time} how transmission time varies with the covertness constraint $\epsilon$. As illustrated in Fig.~\ref{fig_covertness_time}, transmission time decreases with the increase of $\epsilon$ (i.e, a relaxation of the covertness constraint). This is attributed to the UAV’s ability to transmit CI at a higher power when $\epsilon$ is larger, consequently reducing the overall transmission time.

Upon a careful examination of Fig.~\ref{fig_covertness_time}, it is evident that the gap between the transmission time of the proposed transmission protocols and that of the utilizing the center of the minimum circle as UAV location widens as $\epsilon$ decreases, indicating the transmission protocols becomes more significant as the covertness constraint tends to more stringent. Both Fig.~\ref{fig_distance_time} and Fig.~\ref{fig_covertness_time} underscore the efficiency of the transmission protocols in effectively leveraging multicast transmission potential to reduce the transmission time in the UAV system, especially under more stringent covertness constraints.

\section{Conclusion} \label{sec_con}
This paper introduced two transmission protocols, where in the OH transmission protocol, the PSO-based UAV location and transmit power optimization algorithm is proposed, and in the TH transmission, the optimal relay and transmit power optimization exhaustive algorithm is proposed to facilitate the covert multicast in UAV-enabled communication systems. In contrast to the conventional center location-based transmission, our comprehensive numerical results illustrate that the new transmission protocols can remarkably enhance the time efficiency for covert multicast, particularly in systems with stringent covertness constraints. It is anticipated that this work can serve as a promising solution for future UAV-enabled covert multicast communication systems.

\bibliography{bibfile}

\begin{thebibliography}{10}
\providecommand{\url}[1]{#1}
\csname url@samestyle\endcsname
\providecommand{\newblock}{\relax}
\providecommand{\bibinfo}[2]{#2}
\providecommand{\BIBentrySTDinterwordspacing}{\spaceskip=0pt\relax}
\providecommand{\BIBentryALTinterwordstretchfactor}{4}
\providecommand{\BIBentryALTinterwordspacing}{\spaceskip=\fontdimen2\font plus
\BIBentryALTinterwordstretchfactor\fontdimen3\font minus
  \fontdimen4\font\relax}
\providecommand{\BIBforeignlanguage}[2]{{%
\expandafter\ifx\csname l@#1\endcsname\relax
\typeout{** WARNING: IEEEtran.bst: No hyphenation pattern has been}%
\typeout{** loaded for the language `#1'. Using the pattern for}%
\typeout{** the default language instead.}%
\else
\language=\csname l@#1\endcsname
\fi
#2}}
\providecommand{\BIBdecl}{\relax}
\BIBdecl

\bibitem{zeng2019accessing}
Y.~Zeng, Q.~Wu, and R.~Zhang, ``Accessing from the sky: a tutorial on {UAV}
  communications for 5{G} and beyond,'' \emph{Proceedings of the IEEE}, vol.
  107, no.~12, pp. 2327--2375, Dec. 2019.

\bibitem{wu2018cooperative}
H.~Wu, X.~Tao, N.~Zhang, and X.~Shen, ``Cooperative {UAV} cluster-assisted
  terrestrial cellular networks for ubiquitous coverage,'' \emph{IEEE Journal
  on Selected Areas in Communications}, vol.~36, no.~9, pp. 2045--2058, Aug.
  2018.

\bibitem{cui2018robust}
M.~Cui, G.~Zhang, Q.~Wu, and D.~W.~K. Ng, ``Robust trajectory and transmit
  power design for secure {UAV} communications,'' \emph{IEEE Transactions on
  Vehicular Technology}, vol.~67, no.~9, pp. 9042--9046, Jun. 2018.

\bibitem{bash2013limits}
B.~A. Bash, D.~Goeckel, and D.~Towsley, ``Limits of reliable communication with
  low probability of detection on {AWGN} channels,'' \emph{IEEE Journal on
  Selected Areas in Communications}, vol.~31, no.~9, pp. 1921--1930, Aug. 2013.

\bibitem{zhou2019joint}
X.~Zhou, S.~Yan, J.~Hu, J.~Sun, J.~Li, and F.~Shu, ``Joint optimization of a
  {UAV}'s trajectory and transmit power for covert communications,'' \emph{IEEE
  Transactions on Signal Processing}, vol.~67, no.~16, pp. 4276--4290, July
  2019.

\bibitem{rao2022optimal}
H.~Rao, S.~Xiao, S.~Yan, J.~Wang, and W.~Tang, ``Optimal geometric solutions to
  {UAV}-enabled covert communications in line-of-sight scenarios,'' \emph{IEEE
  Transactions on Wireless Communications}, vol.~21, no.~12, pp.
  10\,633--10\,647, Jun. 2022.

\bibitem{jiang2021covert}
X.~Jiang, X.~Chen, J.~Tang, N.~Zhao, X.~Y. Zhang, D.~Niyato, and K.-K. Wong,
  ``Covert communication in {UAV}-assisted air-ground networks,'' \emph{IEEE
  Wireless Communications}, vol.~28, no.~4, pp. 190--197, Mar. 2021.

\bibitem{yan2021optimal}
S.~Yan, S.~V. Hanly, and I.~B. Collings, ``Optimal transmit power and flying
  location for {UAV} covert wireless communications,'' \emph{IEEE Journal on
  Selected Areas in Communications}, vol.~39, no.~11, pp. 3321--3333, Jun.
  2021.

\bibitem{zhou2021three}
X.~Zhou, S.~Yan, D.~W.~K. Ng, and R.~Schober, ``Three-dimensional placement and
  transmit power design for {UAV} covert communications,'' \emph{IEEE
  Transactions on Vehicular Technology}, vol.~70, no.~12, pp. 13\,424--13\,429,
  Oct. 2021.

\bibitem{chen2021uav}
X.~Chen, M.~Sheng, N.~Zhao, W.~Xu, and D.~Niyato, ``{UAV}-relayed covert
  communication towards a flying warden,'' \emph{IEEE Transactions on
  Communications}, vol.~69, no.~11, pp. 7659--7672, Aug. 2021.

\bibitem{zhang2022uav}
R.~Zhang, X.~Chen, M.~Liu, N.~Zhao, X.~Wang, and A.~Nallanathan, ``{UAV} relay
  assisted cooperative jamming for covert communications over rician fading,''
  \emph{IEEE Transactions on Vehicular Technology}, vol.~71, no.~7, pp.
  7936--7941, Apr. 2022.

\bibitem{jiao2022placement}
L.~Jiao, R.~Zhang, M.~Liu, Q.~Hua, N.~Zhao, A.~Nallanathan, and X.~Wang,
  ``Placement optimization of {UAV} relaying for covert communication,''
  \emph{IEEE Transactions on Vehicular Technology}, vol.~71, no.~11, pp.
  12\,327--12\,332, July 2022.

\bibitem{chen2023uav}
X.~Chen, Z.~Chang, M.~Liu, N.~Zhao, T.~H{\"a}m{\"a}l{\"a}inen, and D.~Niyato,
  ``{UAV-IRS} assisted covert communication: Introducing uncertainty via phase
  shifting,'' \emph{IEEE Wireless Communications Letters}, Oct. 2023.

\bibitem{wang2022covert}
{Wang, Chao and Chen, Xinying and An, Jianping and Xiong, Zehui and Xing,
  Chengwen and Zhao, Nan and Niyato, Dusit}, ``Covert communication assisted by
  {UAV-IRS},'' \emph{IEEE Transactions on Communications}, vol.~71, no.~1, pp.
  357--369, Jan. 2022.

\bibitem{qian2023joint}
Y.~Qian, C.~Yang, Z.~Mei, X.~Zhou, L.~Shi, and J.~Li, ``On joint optimization
  of trajectory and phase shift for {IRS-UAV} assisted covert communication
  systems,'' \emph{IEEE Transactions on Vehicular Technology}, Apr. 2023.

\bibitem{bi2023deep}
S.~Bi, L.~Hu, Q.~Liu, J.~Wu, R.~Yang, and L.~Wu, ``Deep reinforcement learning
  for {IRS}-assisted {UAV} covert communications,'' \emph{China
  Communications}, May 2023.

\bibitem{tatar2022aerial}
M.~Tatar~Mamaghani and Y.~Hong, ``Aerial intelligent reflecting surface-enabled
  terahertz covert communications in beyond-{5G} internet of things,''
  \emph{IEEE Internet of Things Journal}, vol.~9, no.~19, pp. 19\,012--19\,033,
  Oct. 2022.

\bibitem{su2023optimal}
Y.~Su, S.~Fu, J.~Si, C.~Xiang, N.~Zhang, and X.~Li, ``Optimal hovering height
  and power allocation for {UAV}-aided {NOMA} covert communication system,''
  \emph{IEEE Wireless Communications Letters}, Feb. 2023.

\bibitem{deng2023joint}
D.~Deng, S.~Dang, X.~Li, D.~W.~K. Ng, and A.~Nallanathan, ``Joint optimization
  for covert communications in {UAV}-assisted {NOMA} networks,'' \emph{IEEE
  Transactions on Vehicular Technology}, Aug. 2023.

\bibitem{sidiropoulos2006transmit}
N.~D. Sidiropoulos, T.~N. Davidson, and Z.-Q. Luo, ``Transmit beamforming for
  physical-layer multicasting,'' \emph{IEEE transactions on signal processing},
  vol.~54, no.~6, pp. 2239--2251, June 2006.

\bibitem{lin2021supporting}
Z.~Lin, M.~Lin, T.~de~Cola, J.-B. Wang, W.-P. Zhu, and J.~Cheng, ``Supporting
  iot with rate-splitting multiple access in satellite and aerial-integrated
  networks,'' \emph{IEEE Internet of Things Journal}, vol.~8, no.~14, pp.
  11\,123--11\,134, Jan. 2021.

\bibitem{shu2019delay}
F.~Shu, T.~Xu, J.~Hu, and S.~Yan, ``Delay-constrained covert communications
  with a full-duplex receiver,'' \emph{IEEE Wireless Communications letters},
  vol.~8, no.~3, pp. 813--816, Jan. 2019.

\bibitem{buffer_TCOM}
J.~He, J.~Liu, W.~Su, Y.~Shen, X.~Jiang, and N.~Shiratori, ``Jamming and link
  selection for joint secrecy/delay guarantees in buffer-aided relay system,''
  \emph{IEEE Transactions on Communications}, vol.~70, no.~8, pp. 5451--5468,
  Aug. 2022.

\bibitem{riihonen2011hybrid}
T.~Riihonen, S.~Werner, and R.~Wichman, ``Hybrid full-duplex/half-duplex
  relaying with transmit power adaptation,'' \emph{IEEE transactions on
  wireless communications}, vol.~10, no.~9, pp. 3074--3085, July, 2011.

\bibitem{hu2015capacity}
Y.~Hu, J.~Gross, and A.~Schmeink, ``On the capacity of relaying with finite
  blocklength,'' \emph{IEEE Transactions on Vehicular Technology}, vol.~65,
  no.~3, pp. 1790--1794, Feb. 2015.

\bibitem{wang2018covert}
J.~Wang, W.~Tang, Q.~Zhu, X.~Li, H.~Rao, and S.~Li, ``Covert communication with
  the help of relay and channel uncertainty,'' \emph{IEEE Wireless
  Communications Letters}, vol.~8, no.~1, pp. 317--320, Sept. 2018.

\bibitem{gao2021covert}
C.~Gao, B.~Yang, X.~Jiang, H.~Inamura, and M.~Fukushi, ``Covert communication
  in relay-assisted iot systems,'' \emph{IEEE Internet of Things Journal},
  vol.~8, no.~8, pp. 6313--6323, Jan. 2021.

\bibitem{xiang2020secure}
Z.~Xiang, W.~Yang, Y.~Cai, J.~Xiong, Z.~Ding, and Y.~Song, ``Secure
  transmission in a {NOMA}-assisted {IOT} network with diversified
  communication requirements,'' \emph{IEEE Internet of Things Journal}, vol.~7,
  no.~11, pp. 11\,157--11\,169, May 2020.

\bibitem{durisi2016toward}
G.~Durisi, T.~Koch, and P.~Popovski, ``Toward massive, ultrareliable, and
  low-latency wireless communication with short packets,'' \emph{Proceedings of
  the IEEE}, vol. 104, no.~9, pp. 1711--1726, Aug. 2016.

\bibitem{lu2021joint}
X.~Lu, W.~Yang, S.~Yan, L.~Tao, and D.~W.~K. Ng, ``Joint packet generation and
  covert communication in delay-intolerant status update systems,'' \emph{IEEE
  Transactions on Vehicular Technology}, vol.~71, no.~2, pp. 2170--2175, Dec.
  2021.

\bibitem{wu2023irs}
Y.~Wu, X.~Chen, M.~Liu, L.~Xu, N.~Zhao, X.~Wang, and D.~W.~K. Ng,
  ``Irs-assisted covert communication with equal and unequal transmit prior
  probabilities,'' \emph{IEEE Transactions on Communications}, Dec. 2023.

\bibitem{sanchez2019distributed}
J.~S{\'a}nchez-Garc{\'\i}a, D.~G. Reina, and S.~Toral, ``A distributed
  {PSO}-based exploration algorithm for a {UAV} network assisting a disaster
  scenario,'' \emph{Future Generation Computer Systems}, vol.~90, pp. 129--148,
  Jan. 2019.

\bibitem{liu2023joint}
Y.~Liu, H.~Wu, and X.~Jiang, ``Joint selection of {FD/HD} and {AF/DF} for
  covert communication in two-hop relay systems,'' \emph{Ad Hoc Networks}, vol.
  148, p. 103207, Sept. 2023.

\bibitem{gradshteyn2014table}
I.~S. Gradshteyn and I.~M. Ryzhik, \emph{Table of integrals, series, and
  products}.\hskip 1em plus 0.5em minus 0.4em\relax Academic press, 2014.

\end{thebibliography}
\bibliographystyle{IEEEtran}

\end{document}